\documentclass[a4paper,UKenglish,cleveref,autoref,table]{lipics-v2021}
\pdfoutput=1
\title{Towards Practical First-Order Model Counting}

\author{Ananth K. Kidambi}{Indian Institute of Technology Bombay, Mumbai, India}{210051002@iitb.ac.in}{}{}
\author{Guramrit Singh}{Indian Institute of Technology Bombay, Mumbai, India}{guramrit@iitb.ac.in}{}{}
\author{Paulius Dilkas}{University of Toronto, Toronto, Canada \and Vector Institute, Toronto, Canada \and \url{https://dilkas.github.io/}}{paulius.dilkas@utoronto.ca}{https://orcid.org/0000-0001-9185-7840}{}
\author{Kuldeep S. Meel}{Georgia Institute of Technology, Atlanta, GA, USA \and University of Toronto, Toronto, Canada \and \url{https://www.cs.toronto.edu/~meel/}}{meel@cs.toronto.edu}{https://orcid.org/0000-0001-9423-5270}{}

\authorrunning{A.\,K. Kidambi, G. Singh, P. Dilkas, and K.\,S. Meel}
\Copyright{Ananth K. Kidambi, Guramrit Singh, Paulius Dilkas, and Kuldeep S. Meel}



\supplementdetails[subcategory={Source Code}]{Software}{https://doi.org/10.5281/zenodo.15583934}

\funding{This research was funded in part by the Natural Sciences and
  Engineering Research Council of Canada (NSERC), funding reference number
  RGPIN-2024-05956, and the Digital Research Alliance of Canada
  (\texttt{\href{https://alliancecan.ca/}{alliancecan.ca}}).}

\acknowledgements{The first two authors contributed equally. Part of the
  research was conducted while all authors were at the National University of
  Singapore. We express our gratitude to the anonymous reviewers for their
  valuable comments and suggestions.}

\nolinenumbers
\hideLIPIcs

\usepackage[linesnumbered,ruled,vlined]{algorithm2e}
\usepackage{siunitx}
\usepackage{tikz}
\usepackage{mathtools}
\usepackage{booktabs}
\usepackage{pifont}

\usetikzlibrary{arrows.meta}
\usetikzlibrary{fit}
\usetikzlibrary{positioning}
\usetikzlibrary{shapes.misc}

\newcommand{\expr}{\mathtt{expr}}
\newcommand{\Cranetwo}{\textsc{Crane2}}
\newcommand{\Cranebfs}{\textsc{Crane2-BFS}}
\newcommand{\Cranegreedy}{\textsc{Crane2-Greedy}}
\newcommand{\friends}{\emph{Friends \& Smokers}}
\newcommand{\functions}{\emph{Functions}}
\newcommand{\bijections}{\emph{Bijections}}

\newcommand{\cmark}{\ding{51}}
\newcommand{\xmark}{\ding{55}}
\newcommand{\Ctwo}{$\mathsf{C}^{2}$}
\newcommand{\FO}{$\mathsf{FO}$}
\newcommand{\UFO}{$\mathsf{UFO}^{2} + \mathsf{CC}$}

\SetKwProg{Fn}{Function}{:}{}
\SetKwData{Cache}{Cache}

\SetKwFunction{CompileWithBaseCases}{CompileWithBaseCases}
\SetKwFunction{Compile}{{\normalfont \textsc{Crane}}}
\SetKwFunction{Propagate}{Propagate}
\SetKwFunction{FindBaseCases}{FindBaseCases}
\SetKwFunction{Simplify}{Simplify}
\SetKwFunction{Main}{Main}
\SetKwFunction{Input}{ParseCommandLineArguments}

\DeclareMathOperator{\Doms}{Doms}

\newtheorem{assumption}[theorem]{Assumption}

\crefalias{enumi}{type}
\crefname{type}{Type}{Types}
\creflabelformat{type}{#2\textup{#1}#3}
\crefalias{clause}{equation}
\crefname{clause}{clause}{clauses}
\creflabelformat{clause}{#2\textup{(#1)}#3}
\crefalias{formula}{equation}
\crefname{formula}{sentence}{sentences}
\creflabelformat{formula}{#2\textup{(#1)}#3}

\begin{document}

\maketitle

\begin{abstract}
  First-order model counting (FOMC) is the problem of counting the number of
  models of a sentence in first-order logic. Since lifted inference techniques
  rely on reductions to variants of FOMC, the design of scalable methods for
  FOMC has attracted attention from both theoreticians and practitioners over
  the past decade. Recently, a new approach based on first-order knowledge
  compilation was proposed. This approach, called \textsc{Crane}, instead of
  simply providing the final count, generates definitions of (possibly
  recursive) functions that can be evaluated with different arguments to compute
  the model count for any domain size. However, this approach is not fully
  automated, as it requires manual evaluation of the constructed functions. The
  primary contribution of this work is a fully automated compilation algorithm,
  called \Cranetwo{}, which transforms the function definitions into C++ code
  equipped with arbitrary-precision arithmetic. These additions allow the new
  FOMC algorithm to scale to domain sizes over \num{500000} times larger than
  the current state of the art, as demonstrated through experimental results.
\end{abstract}

\section{Introduction}

\emph{First-order model counting} (FOMC) is the task of determining the number
of models for a sentence in first-order logic over a specified domain. The
weighted variant, WFOMC, computes the total weight of these models, linking
logical reasoning with probabilistic
frameworks~\cite{DBLP:conf/ijcai/BroeckTMDR11}. It builds upon earlier efforts
in weighted model counting for propositional
logic~\cite{DBLP:journals/ai/ChaviraD08} and broader attempts to bridge logic
and
probability~\cite{DBLP:journals/ai/Nilsson86,novak2012mathematical,vsaletic2024graded}.
WFOMC is central to \emph{lifted inference}, which enhances the efficiency of
probabilistic calculations by exploiting
symmetries~\cite{DBLP:conf/ecai/Kersting12}. Lifted inference continues to
advance, with applications extending to constraint satisfaction
problems~\cite{DBLP:journals/jair/TotisDRK23} and probabilistic answer set
programming~\cite{DBLP:journals/ijar/AzzoliniR23}. Moreover, WFOMC has proven
effective at reasoning over probabilistic
databases~\cite{DBLP:journals/debu/GribkoffSB14} and probabilistic logic
programs~\cite{DBLP:journals/ijar/RiguzziBZCL17}. FOMC algorithms have also
facilitated breakthroughs in discovering integer
sequences~\cite{DBLP:conf/ijcai/SvatosJT0K23} and developing recurrence
relations for these sequences~\cite{DBLP:conf/kr/DilkasB23}. Recently, these
algorithms have been extended to perform sampling
tasks~\cite{DBLP:journals/ai/WangPWK24}.

The complexity of FOMC is generally measured by \emph{data complexity}, with a
sentence classified as \emph{liftable} if it admits a polynomial-time solution
relative to the domain size~\cite{DBLP:conf/starai/JaegerB12}. While all
sentences with up to two variables are known to be
liftable~\cite{DBLP:conf/nips/Broeck11,DBLP:conf/kr/BroeckMD14}, Beame et
al.~\cite{DBLP:conf/pods/BeameBGS15} demonstrated that liftability does not
extend to all sentences, identifying an unliftable sentence with three
variables. Recent work has extended the liftable fragment with additional
axioms~\cite{DBLP:conf/aaai/TothK23,DBLP:journals/ai/BremenK23} and counting
quantifiers~\cite{DBLP:journals/jair/Kuzelka21}, expanding our understanding of
liftability.

FOMC algorithms are diverse, with approaches ranging from \emph{first-order
  knowledge compilation} (FOKC) to cell
enumeration~\cite{DBLP:conf/uai/BremenK21}, local
search~\cite{DBLP:journals/pvldb/NiuRDS11}, and Monte Carlo
sampling~\cite{DBLP:journals/cacm/GogateD16}. FOKC-based algorithms are
particularly prominent, transforming sentences into structured representations
such as circuits or graphs. Even when multiple algorithms can solve the same
instance, FOKC algorithms are known to find polynomial-time solutions where the
polynomial has a lower degree than other
approaches~\cite{DBLP:conf/kr/DilkasB23}. The recently developed ability of a
FOKC algorithm to formulate solutions in terms of recursive
functions~\cite{DBLP:conf/kr/DilkasB23} is also noteworthy since the only other
proposed alternative is to guess recursive
relations~\cite{DBLP:conf/ilp/BarvinekB0ZK21}. Notable examples of FOKC
algorithms include \textsc{ForcLift}~\cite{DBLP:conf/ijcai/BroeckTMDR11} and its
extension \textsc{Crane}~\cite{DBLP:conf/kr/DilkasB23}.

The \textsc{Crane} algorithm marked a significant step forward, expanding the
range of sentences handled by FOMC algorithms. However, it had notable
limitations: it required manual evaluation of function definitions to compute
model counts and introduced recursive functions without proper base cases,
making it more difficult to use. To address these shortcomings, we present
\Cranetwo{}, a fully automated FOMC algorithm that overcomes the constraints of
its predecessor. \Cranetwo{} can handle domain sizes over \num{500000} times
larger than those of previous algorithms and simplifies the user experience by
automatically managing base cases and compiling function definitions into
efficient C++ programs.

The paper is structured as follows. \Cref{sec:preliminaries} covers the
necessary preliminaries, including notation, terminology, and background
information about (W)FOMC and FOKC\@. \Cref{sec:main} describes all aspects of
\Cranetwo{}:
\begin{itemize}
  \item its overall structure and interactions with \textsc{Crane},
  \item how we identify a sufficient set of base cases for a recursive function,
  \item how we construct the sentence that describes each base case, and
  \item how we compile function definitions into a C++ program with caching
        mechanisms that ensure efficiency.
\end{itemize}
Then, \cref{sec:experiments} presents our experimental results, demonstrating
\Cranetwo{}'s performance compared to other FOMC algorithms. Finally,
\cref{sec:conclusion} concludes the paper by discussing promising avenues for
future work.

\section{Preliminaries}\label{sec:preliminaries}

We begin this section by describing some notation that we will use throughout
the paper. Then, in \cref{sec:logic,sec:fomc}, we introduce the basic
terminology of first-order logic and formally define (W)FOMC\@.
\Cref{sec:threelogics} provides a brief overview of the different types of
first-order logics commonly used in FOMC literature. \Cref{sec:crane} outlines
the principles of FOKC, particularly in the context of \textsc{Crane}. Finally,
in \cref{sec:algebra}, we introduce the algebraic terminology used to describe
the output of \textsc{Crane}, i.e., functions and equations that define them.

\paragraph*{Notation}
We use $\mathbb{N}_{0}$ to represent the set of non-negative integers. In both
algebra and logic, we write
$S[x_{1} \mapsto y_{1}, x_{2} \mapsto y_{2}, \dots, x_{n} \mapsto y_{n}]$ to
denote the application of a \emph{substitution} to an expression $S$, where all
instances of $x_{i}$ are replaced with $y_{i}$ for all $i = 1, \dots, n$.
Additionally, for any variable $n$ and $a, b \in \mathbb{N}_{0}$, let
$[a \le n \le b] \coloneqq \begin{cases}
  1 & \text{if $a \le n \le b$} \\
  0 & \text{otherwise}
\end{cases}$.

\subsection{First-Order Logic}\label{sec:logic}

This section reviews the basic concepts of first-order logic used in FOKC
algorithms. We focus specifically on the format used internally by
\textsc{ForcLift} and its descendants. See \cref{sec:preprocessing} for a brief
overview of how \Cranetwo{} transforms an arbitrary sentence in first-order
logic into this internal format.

A \emph{term} can be either a variable or a constant. An \emph{atom} can be
either $P(t_{1}, \dots, t_{m})$ (i.e., $P(\mathbf{t})$) for some predicate $P$
and terms $t_{1}, \dots, t_{m}$, or $x=y$ for some terms $x$ and $y$. The
\emph{arity} of a predicate is the number of arguments it takes, i.e., $m$ in
the case of the predicate $P$ mentioned above. We write $P/m$ to denote a
predicate along with its arity. A \emph{literal} can be either an atom (i.e., a
\emph{positive} literal) or its negation (i.e., a \emph{negative} literal). An
atom containing no variables, only constants, is called \emph{ground}. A
\emph{clause} is of the form $\forall x_{1} \in \Delta_{1}\text{.
}\forall x_{2} \in \Delta_{2}\dots\text{ }\forall x_{n} \in \Delta_{n}\text{.
}\phi(x_{1}, x_{2}, \dots, x_{n})$, where $\phi$ is a disjunction of literals
that contain only the variables $x_{1}, \dots, x_{n}$ (and any constants). We
say that a clause is a \emph{(positive) unit clause} if there is only one
literal with a predicate, and it is a positive literal. Finally, a
\emph{sentence} is a conjunction of clauses. Throughout the paper, we will use
set-theoretic notation, interpreting a sentence as a set of clauses and a clause
as a set of literals.

\begin{remark*}
  Conforming to previous work~\cite{DBLP:conf/ijcai/BroeckTMDR11}, the
  definition of a clause includes universal quantifiers for all its variables.
  While it is possible to rewrite the entire sentence with all quantifiers at
  the front, the format we describe has proven convenient for practical use.
\end{remark*}

\subsection{First-Order Model Counting}\label{sec:fomc}

In this section, we will formally define FOMC and its weighted variant. Although
this work focuses on FOMC, computing the FOMC using \Cranetwo{} requires using
WFOMC for sentences with existential quantifiers. For such sentences,
preprocessing (described in \cref{sec:preprocessing}) introduces predicates with
non-unary weights that must be accounted for to compute the correct model count.

\begin{definition}[Structure, model]\label{def:model}
  Let $\phi$ be a sentence. For each predicate $P/n$ in $\phi$, let
  ${(\Delta_{i}^{P})}_{i=1}^{n}$ be a list of the corresponding domains. Let
  $\sigma$ be a map from the domains of $\phi$ to their interpretations as
  finite sets, ensuring the sets are pairwise disjoint and contain the
  corresponding constants from $\phi$. A \emph{structure} of $\phi$ is a set $M$
  of ground literals. It is defined by adding to $M$ either $P(\mathbf{t})$ or
  $\neg P(\mathbf{t})$ for every predicate $P/n$ in $\phi$ and every $n$-tuple
  $\mathbf{t} \in \prod_{i=1}^{n} \sigma(\Delta_{i}^{P})$. (Here, $\prod$ is the
  usual Cartesian product.) A structure is a \emph{model} if it makes $\phi$
  valid.
\end{definition}

\begin{example}[Counting bijections]\label{example:bijections}
  Let us consider the following sentence (previously examined by Dilkas and
  Belle~\cite{DBLP:conf/kr/DilkasB23}) that defines predicate $P$ as a bijection
  between two domains $\Gamma$ and $\Delta$:
  \begin{equation}\label[formula]{eq:bijections}
    \begin{gathered}
      (\forall x \in \Gamma\text{. }\exists y \in \Delta\text{. }P(x, y))\land{}\\
      (\forall y \in \Delta\text{. }\exists x \in \Gamma\text{. }P(x, y))\land{}\\
      (\forall x \in \Gamma\text{. }\forall y, z \in \Delta\text{. }P(x, y) \land P(x, z) \Rightarrow y = z)\land{}\\
      (\forall x, z \in \Gamma\text{. }\forall y \in \Delta\text{. }P(x, y) \land P(z, y) \Rightarrow x = z).
    \end{gathered}
  \end{equation}
  Let $\sigma$ be defined as $\sigma(\Gamma) \coloneqq \{\, 1, 2\,\}$ and
  $\sigma(\Delta) \coloneqq \{\,a, b\,\}$. \Cref{eq:bijections} has two models:
  \[
    \{\, P(1, a), P(2, b), \neg P(1, b), \neg P(2, a) \,\} \qquad \text{and} \qquad \{\, P(1, b), P(2, a), \neg P(1, a), \neg P(2, b) \,\}.
  \]
\end{example}

\begin{definition}[WFOMC instance]\label{def:instance}
  A \emph{WFOMC instance} comprises:
  \begin{itemize}
    \item a sentence $\phi$,
    \item two (rational) \emph{weights} $w^{+}(P)$ and $w^{-}(P)$ assigned to
          each predicate $P$ in $\phi$, and
    \item $\sigma$ as described in \cref{def:model}.
  \end{itemize}
  Unless specified otherwise, we assume all weights are equal to 1.
\end{definition}

\begin{definition}[WFOMC~\cite{DBLP:conf/ijcai/BroeckTMDR11}]
  Given a WFOMC instance $(\phi, w^{+}, w^{-}, \sigma)$ as in
  \cref{def:instance}, the \emph{(symmetric) weighted first-order model count}
  (WFOMC) of $\phi$ is
  \begin{equation}\label{eq:wfomc}
    \sum_{M \models \phi} \prod_{P(\mathbf{t}) \in M} w^{+}(P) \prod_{\neg P(\mathbf{t}) \in M} w^{-}(P),
  \end{equation}
  where the sum is over all models of $\phi$.
\end{definition}

\subsection{The Three Logics of FOMC}\label{sec:threelogics}

\begin{table}[t]
  \centering
  \caption{A comparison of the three logics used in FOMC\@. The
    2\textsuperscript{nd}--5\textsuperscript{th} columns refer to the number of
    sorts, support for constants, the maximum number of variables, and supported
    quantifiers, respectively. The last column lists supported atoms in addition
    to those of the form $P(\mathbf{t})$ for a predicate $P/n$ and an $n$-tuple
    of terms $\mathbf{t}$. Here, $k$ and $m$ are non-negative integers, where
    $m$ depends on the domain size, $P$ is a predicate, and $x$ and $y$ are
    terms.}\label{tbl:logics}
  \begin{tabular}{llclll}
    \toprule
    Logic & Sorts & Constants & Variables & Quantifiers & Add.\ atoms\\
    \midrule
    \rowcolor{gray!10}
    \FO & one or more & \cmark & unlimited & $\forall$, $\exists$ & $x = y$\\
    \Ctwo & \multirow{2}{*}{one} & \multirow{2}{*}{\xmark} & \multirow{2}{*}{two} & $\forall$, $\exists$, $\exists^{= k}$, $\exists^{\le k}$, $\exists^{\ge k}$ & ---\\
    \UFO & & & & $\forall$ & $|P| = m$\\
    \bottomrule
  \end{tabular}
\end{table}

FOMC commonly utilises three types of first-order logic: \FO{}, \Ctwo{}, and
\UFO{}. \Cref{tbl:logics} summarises the key differences among them. \FO{} is
the input format for \textsc{ForcLift} and its extensions \textsc{Crane} and
\Cranetwo{}. \Ctwo{} is often used in the literature on \textsc{FastWFOMC} and
related methods~\cite{DBLP:journals/jair/Kuzelka21,DBLP:conf/aaai/MalhotraS22}.
(Note that no algorithm accepts \Ctwo{} as input.) Finally, \UFO{} is the input
format supported by the most recent implementation of
\textsc{FastWFOMC}~\cite{DBLP:conf/kr/TothK24}. All three logics are
function-free, and domains are always assumed to be finite. As usual, we
presuppose the \emph{unique name assumption}, which states that two constants
are equal if and only if they are the same constant~\cite{DBLP:books/aw/RN2020}.

In \FO{}, each term has a designated \emph{sort}, and each predicate $P/n$
corresponds to a sequence of $n$ sorts. Each sort has its corresponding domain.
These assignments to sorts are typically left implicit and follow from the
quantifiers, e.g., $\forall x,y \in \Delta\text{. }P(x, y)$ implies that the
variables $x$ and $y$ have the same sort. On the other hand,
$\forall x \in \Delta\text{. }\forall y \in \Gamma\text{. } P(x, y)$ implies
that $x$ and $y$ have different sorts, and it would be improper to write, for
example, $\forall x \in \Delta\text{. }\forall y \in \Gamma\text{.
} P(x, y) \lor x = y$. \FO{} is also the only logic to support constants,
sentences with more than two variables, and the equality predicate. While we do
not explicitly refer to sorts in the rest of the paper, the many-sorted nature
of \FO{} is paramount to the algorithms presented therein.

\begin{remark*}
  In the case of \textsc{ForcLift} and its extensions, support for a sentence as
  valid input does not imply that the algorithm can compile the sentence into a
  circuit or graph suitable for lifted model counting. However,
  \textsc{ForcLift} compilation always succeeds on any \FO{} sentence without
  constants and with at most two
  variables~\cite{DBLP:conf/nips/Broeck11,DBLP:conf/kr/BroeckMD14}.
\end{remark*}

Compared to \FO{}, \Ctwo{} and \UFO{} lack support for constants, the equality
predicate, multiple domains, and sentences with more than two variables. The
advantage that \Ctwo{} brings over \FO{} is the inclusion of \emph{counting
  quantifiers}. That is, alongside $\forall$ and $\exists$, \Ctwo{} supports
$\exists^{=k}$, $\exists^{\le k}$, and $\exists^{\ge k}$ for any positive
integer $k$. For example, $\exists^{=1} x\text{. }\phi(x)$ means that there
exists \emph{exactly one} $x$ such that $\phi(x)$, and $\exists^{\le 2} x\text{.
}\phi(x)$ means that there exist \emph{at most two} such $x$. \UFO{}, on the
other hand, does not support any existential quantifiers but instead
incorporates \emph{(equality) cardinality constraints}. For example, $|P| = 3$
constrains all models to have \emph{precisely three positive literals with the
  predicate $P$}. Except when stated otherwise, we will write all example
sentences in \FO{}.

\subsection{Crane and First-Order Knowledge Compilation}\label{sec:crane}

As our work builds on \textsc{Crane}, in this section, we will briefly outline
the steps \textsc{Crane} goes through to compile a sentence into a set of
function definitions. We divide the inner workings of the algorithm into two
stages: preprocessing and compilation.

\subsubsection{Preprocessing}\label{sec:preprocessing}

This stage transforms an arbitrary sentence into the format described in
\cref{sec:logic}, primarily by eliminating existential quantifiers via a process
called \emph{Skolemization}~\cite{DBLP:conf/kr/BroeckMD14}. For example, the
first conjunct of \cref{eq:bijections}, i.e.,
\begin{equation}\label[formula]{eq:skolemizationinitial}
  \forall x \in \Gamma\text{. }\exists y \in \Delta\text{. } P(x, y)
\end{equation}
is transformed into
\begin{equation}\label[formula]{eq:skolemization}
  \begin{aligned}
    (\forall x \in \Gamma\text{. } &Z(x))\land{}\\
    (\forall x \in \Gamma\text{. } \forall y \in \Delta\text{. } &Z(x) \lor \neg P(x, y))\land{}\\
    (\forall x \in \Gamma\text{. } &S(x) \lor Z(x))\land{}\\
    (\forall x \in \Gamma\text{. } \forall y \in \Delta\text{. } &S(x) \lor \neg P(x, y)),
  \end{aligned}
\end{equation}
where $Z/1$ and $S/1$ are two new predicates with $w^{-}(S) = -1$. The idea
behind Skolemization is that the output sentence has more models, but for every
additional model with a positive weight, there is an equivalent model with a
negative weight. Hence, the WFOMC of
\cref{eq:skolemizationinitial,eq:skolemization} is the same. Moreover, note that
in this example, preprocessing introduces some redundancy, i.e., the second and
third clauses of \cref{eq:skolemization} could be removed. These redundancies
are handled at the beginning of the compilation process described in
\cref{sec:compilation}.

To illustrate how Skolemization preserves the WFOMC, we let $\Gamma = \{\,a\,\}$
and $\Delta = \{\,1, 2\,\}$ and assume that $w^+(P) = w^-(P) = 1$. Here, for
conciseness, we omit negative literals when writing models.
\Cref{eq:skolemizationinitial} has three models: $\{\,P(a, 1)\,\}$,
$\{\,P(a, 2)\,\}$, and $\{\,P(a, 1), P(a, 2)\,\}$, all with weight 1, so the
WFOMC is 3. \Cref{eq:skolemization} has the following models:
\begin{itemize}
  \item $\{\,Z(a)\,\}$,
  \item $\{\,Z(a), S(a)\,\}$,
  \item $\{\,Z(a), S(a), P(a, 1)\,\}$,
  \item $\{\,Z(a), S(a), P(a, 2)\,\}$,
  \item $\{\,Z(a), S(a), P(a, 1), P(a, 2)\,\}$.
\end{itemize}
The first model has a weight of -1 because $w^-(S) = -1$, while all the other
models have a weight of 1. Hence, the WFOMC is $4 - 1 = 3$.

\subsubsection{Compilation}\label{sec:compilation}

After preprocessing, \textsc{Crane} compiles the sentence into the triple
$(\mathcal{E}, \mathcal{F}, \mathcal{D})$, where $\mathcal{E}$ is the set of
equations, and $\mathcal{F}$ and $\mathcal{D}$ are auxiliary maps. $\mathcal{F}$
maps function names to sentences in such a way that the evaluation of the former
corresponds to the FOMC of the latter. $\mathcal{D}$ maps function names and
argument indices to domains. The equations in $\mathcal{E}$ can define an
arbitrary number of functions, one of which (which we will always denote as $f$)
represents the solution to the FOMC problem. Computing the FOMC for specific
domain sizes involves invoking $f$ with those sizes as inputs. $\mathcal{D}$
records this correspondence between function arguments and domains.

\begin{example}\label{example:solution}
  \textsc{Crane} compiles \cref{eq:bijections} for bijection counting into
  \begin{align*}
    \mathcal{E} &= \left\{\,\begin{aligned}f(m, n) &= \sum_{l=0}^{n} \binom{n}{l}{(-1)}^{n-l}g(l, m),\\ g(l, m) &= \sum_{k=0}^{m}[0 \le k \le 1]\binom{m}{k}g(l-1, m-k)\end{aligned}\,\right\};\\
    \mathcal{D} &= \{\, (f, 1) \mapsto \Gamma, (f, 2) \mapsto \Delta, (g, 1) \mapsto \Delta^{\top}, (g, 2) \mapsto \Gamma \,\},
  \end{align*}
  where $\Delta^{\top}$ is a newly introduced domain. (We omit the definition of
  $\mathcal{F}$ as the sentences can become quite verbose.) To compute the
  number of bijections between two sets of cardinality 3, one would evaluate
  $f(3, 3)$; however, the definition of $g$ is incomplete: $g$ is a recursive
  function presented without any base cases. $\mathcal{D}$ encodes that in
  $f(m, n)$, $m$ and $n$ represent $|\Gamma|$ and $|\Delta|$, respectively.
  Similarly, in $g(l, m)$, $l$ represents $|\Delta^{\top}|$, and $m$ represents
  $|\Gamma|$.
\end{example}

Compilation is performed primarily by applying \emph{(compilation) rules} to
sentences. \textsc{Crane} has two modes depending on the selection process for
compilation rules when multiple alternatives are available. The first option is
to use greedy search: there is a list of rules, and the first applicable rule is
the one that gets used, disregarding all the others. The second option is to use
a combination of greedy and \emph{breadth-first search} (BFS). In this approach,
we classify each compilation rule as greedy or non-greedy. Greedy rules are
applied as soon as possible at any stage of the compilation process, while BFS
is executed over all applicable non-greedy rules, identifying the solution that
necessitates the smallest number of non-greedy rules.

\subsection{Algebra}\label{sec:algebra}

In this paper, we use both logical and algebraic constructs. While the rest of
\cref{sec:preliminaries} focused on the former, this section describes the
latter. We write $\expr{}$ for an arbitrary algebraic expression. In the context
of algebra, a \emph{constant} is a non-negative integer. Likewise, a
\emph{variable} can either be a parameter of a function or a variable introduced
through summation, such as $i$ in the expression $\sum_{i=1}^{n} \expr$. A
\emph{function call} is $f(x_{1}, \dots, x_{n})$ (or $f(\mathbf{x})$ for short),
where $f$ is an $n$-ary function, and each $x_{i}$ is an algebraic expression
consisting of variables and constants. A (function) \emph{signature} is a
function call that contains only variables. Given two function calls,
$f(\mathbf{x})$ and $f(\mathbf{y})$, we say that $f(\mathbf{y})$ \emph{matches}
$f(\mathbf{x})$ if $x_{i} = y_{i}$ whenever $x_{i}, y_{i} \in \mathbb{N}_{0}$.
An \emph{equation} is $f(\mathbf{x}) = \expr{}$, where $f(\mathbf{x})$ is a
function call.

\begin{definition}[Base case]\label{def:basecase}
  Let $f(\mathbf{x})$ be a function call where each $x_{i}$ is either a constant
  or a variable. Then the function call $f(\mathbf{y})$ is a \emph{base case} of
  $f(\mathbf{x})$ if $f(\mathbf{y}) = f(\mathbf{x})\sigma$, where $\sigma$ is a
  substitution that replaces one or more \emph{variable} $x_{i}$'s with a
  constant while leaving constants unchanged.
\end{definition}

\begin{example}
  In the equation $f(m, n) = f(m-1, n) + nf(m-1, n-1)$, the only constant is
  $1$, and the variables are $m$ and $n$. The equation contains three function
  calls: one on the left-hand side (LHS) and two on the right-hand side (RHS).
  The function call on the LHS is a signature. Function calls such as $f(4, n)$,
  $f(m, 0)$, and $f(8, 1)$ are all considered base cases of $f(m, n)$ (only some
  of which are useful).
\end{example}

\section{Technical Contributions}\label{sec:main}

\begin{figure}[t]
  \centering
  \begin{tikzpicture}
    \node[anchor=west] at (-1, 0) (formula) {$\phi$};
    \node[draw,rounded rectangle] at (3, 0) (compilewithbasecases) {\CompileWithBaseCases};
    \node[draw,rounded rectangle] at (9, 0) (compilation) {Compile to C++};

    \node[draw,rounded rectangle,dashed] at (9, -3) (cpp) {C++ code};
    \node[anchor=west] at (-1, -3) (sizes) {Domain sizes};
    \node at (12, -3) (count) {Answer};

    \node[draw,rounded rectangle] at (3, -2) (findbasecases) {\FindBaseCases};
    \node[draw,rounded rectangle,left = 0.1cm of findbasecases] (crane) {\Compile};
    \node[draw,rounded rectangle,right = 0.1cm of findbasecases] (propagate) {\Propagate};
    \node[draw,rounded rectangle,right = 0.1cm of propagate] (simplify) {\Simplify};

    \node[draw,fit={(compilewithbasecases) (compilation) (crane) (findbasecases) (propagate)},inner ysep=7pt,yshift=5pt] {};
    \node at (0.7, 0.5) {\Cranetwo};

    \draw[-Latex] (formula) -- (compilewithbasecases);
    \draw[-Latex] (compilewithbasecases) -- node[above] {$\mathcal{E}$} (compilation);
    \draw[-Latex] (compilation) -- (cpp);
    \draw[-Latex] (sizes) -- (cpp);
    \draw[-Latex] (cpp) -- (count);

    \draw[-Latex,dashed] (compilewithbasecases) -- node[midway,left] {uses} (crane);
    \draw[-Latex,dashed] (compilewithbasecases) -- node[midway,left] {uses} (findbasecases);
    \draw[-Latex,dashed] (compilewithbasecases) -- node[midway,left] {uses} (propagate);
    \draw[-Latex,dashed] (compilewithbasecases) -- node[midway,right] {uses} (simplify);
  \end{tikzpicture}
  \caption[]{Using \Cranetwo{} to compute the model count of a sentence $\phi$.
    First, \Cranetwo{} compiles $\phi$ into a set of equations, which form the
    basis for creating a C++ program. Executing the program with different
    command line arguments calculates the model count of $\phi$ for various
    domain sizes. To accomplish this, the \CompileWithBaseCases procedure
    employs \textsc{Crane}, algebraic simplification techniques (denoted as
    \Simplify), and two other auxiliary procedures. }\label{fig:overview}
\end{figure}
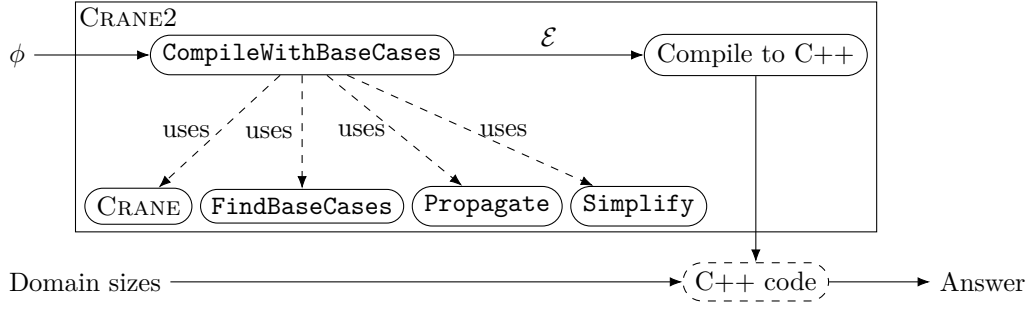

\Cref{fig:overview} provides an overview of \Cranetwo{}'s workflow. We will
briefly describe and motivate each procedure before going into more detail in
the corresponding subsection.

\CompileWithBaseCases (see \cref{sec:completing}), the core procedure of
\Cranetwo{}, is responsible for completing the function definitions produced by
\textsc{Crane} with the necessary base cases. To do so, it may recursively call
itself (and \textsc{Crane}) on other sentences. We prove that the number of such
recursive calls is upper bound by the number of domains.

\cref{sec:completing} also describes the \Simplify procedure for algebraic
simplification. It is crucial for simplifying, e.g., a sum of $n$ terms, only
some of which are non-zero. More generally, the equations returned by
\textsc{Crane} often benefit from easily detectable algebraic simplifications
such as $0 \cdot \text{anything} = 0$ and $\text{anything}^{0} = 1$. Note that
the latter includes $0^{0} = 1$, as 1 is the multiplicative identity.

\FindBaseCases (described in \cref{sec:identifying}) inspects a set of equations
to identify a sufficient set of base cases for a given set of equations. We
prove that the returned set of base cases ensures that the evaluation of the
resulting function definitions will never get stuck in an infinite loop.

\Cref{sec:propagating} introduces the \Propagate procedure, which takes a
sentence $\phi$, a domain $\Delta$, and $n \in \mathbb{N}_{0}$. It returns
$\phi$ transformed under the assumption that $|\Delta| = n$, with $n$ new
constants added and all variables quantified over $\Delta$ eliminated. For
example, when computing a base case such as $f(0, y)$, \Propagate will
significantly simplify $\phi$ with the assumption that the domain associated
with the first parameter of $f$ (i.e., $\mathcal{D}(f, 1)$) is empty.
\CompileWithBaseCases{\Propagate{$\phi$, $\mathcal{D}(f, 1)$, $0$}} will then
return the equations for the base case $f(0, y)$.

\Cref{sec:smoothing} describes a new \emph{smoothing} procedure that ensures
\Propagate preserves the correct model count. Smoothing is a well-known
technique in knowledge compilation algorithms for propositional model
counting~\cite{darwiche2001tractable}. Although FOMC algorithms have used
smoothing before~\cite{DBLP:conf/ijcai/BroeckTMDR11}, our setting requires a
novel approach.

With the help of other procedures outlined above, \CompileWithBaseCases returns
a set of equations that fully cover the base cases of all recursive functions.
While these equations can be intriguing and valuable in their own right, users
of FOMC algorithms typically expect a numerical answer. Thus, \cref{sec:cpp}
describes how \Cranetwo{} compiles these equations into a C++ program that one
can execute with different command-line arguments to compute the model count for
various combinations of domain sizes.

\subsection{Completing the Definitions of Functions}\label{sec:completing}

\begin{algorithm}[t]
  \caption{\protect\CompileWithBaseCases{$\phi$}}\label{alg:compilewithbasecases}
  \KwIn{sentence $\phi$}
  \KwOut{set $\mathcal{E}$ of equations}
  $(\mathcal{E}, \mathcal{F}, \mathcal{D}) \gets \Compile{$\phi$}$\;
  $\mathcal{E} \gets \Simplify{$\mathcal{E}$}$\;\label{line:second}
  \ForEach{base case $f(\mathbf{x}) \in \FindBaseCases{$\mathcal{E}$}$}{\label{line:loopy}
    $\psi \gets \mathcal{F}(f)$\;
    \lForEach{index $i$ such that $x_{i} \in \mathbb{N}_{0}$}{$\psi \gets \Propagate{$\psi$, $\mathcal{D}(f, i)$, $x_i$}$}\label{line:propagate}
    $\mathcal{E} \gets \mathcal{E} \cup \CompileWithBaseCases{$\psi$}$\;\label{line:final}
  }
\end{algorithm}

\Cref{alg:compilewithbasecases} presents our overall approach for compiling a
sentence into equations that include the necessary base cases. First, we use
\textsc{Crane} to compile the sentence into three components: $\mathcal{E}$,
$\mathcal{F}$, and $\mathcal{D}$ (as described in \cref{sec:compilation}). After
some algebraic simplifications (described below), the algorithm passes
$\mathcal{E}$ to the \FindBaseCases procedure (see \cref{sec:identifying}). For
each base case $f(\mathbf{x})$, we retrieve the sentence $\mathcal{F}(f)$
associated with the function name $f$ and simplify it using the \Propagate
procedure (explained in detail in \cref{sec:propagating}). We do this by
iterating over all indices of $\mathbf{x}$, where $x_{i}$ is a constant, and
using \Propagate to simplify $\psi$ by assuming that domain $\mathcal{D}(f, i)$
has size $x_{i}$. Finally, on \autoref{line:final}, \CompileWithBaseCases
recurses on these simplified sentences and adds the resulting base case
equations to $\mathcal{E}$.

\paragraph*{Simplify}
The main responsibility of the \Simplify procedure is to handle the algebraic
pattern $\sum_{m=0}^{n}[a \le m \le b] f(m)$. Here, $n$ is a variable, $a$ and
$b$ are constants, and $f$ is an expression that may depend on $m$. \Simplify
transforms this pattern into $f(a) + f(a+1) + \cdots + f(\min\{\, n, b \,\})$.

\begin{example}\label{example:bijectionstwo}
  {\par\sloppy
  We return to the bijection-counting problem from \cref{example:bijections} and
  its initial solution described in \cref{example:solution}. \Simplify
  transforms
  \[
    g(l, m) = \sum_{k=0}^{m}[0 \le k \le 1]\binom{m}{k}g(l-1, m-k)
  \]
  into
  \[
    g(l, m) = g(l-1, m) + mg(l-1, m-1).
  \]
  Then \FindBaseCases identifies two base cases: $g(0, m)$ and $g(l, 0)$. In
  both cases, \CompileWithBaseCases recurses on the sentence $\mathcal{F}(g)$
  simplified by assuming that one of the domains is empty. In the first case, we
  recurse on the sentence $\forall x \in \Gamma\text{. }S(x) \lor \neg S(x)$,
  where $S$ is a predicate introduced by preprocessing with weights
  $w^{+}(S) = 1$ and $w^{-}(S) = -1$. Hence, we obtain the base case
  $g(0, m) = 0^{m}$. In the case of $g(l, 0)$, \Propagate{$\psi$, $\Gamma$, $0$}
  returns an empty sentence, resulting in $g(l, 0) = 1$. While these base cases
  overlap when $l = m = 0$, they remain consistent since $0^{0} = 1$.
  \par}
\end{example}

Generally, let $\phi$ be a sentence with two domains $\Gamma$ and $\Delta$, and
let $n, m \in \mathbb{N}_{0}$. Then the FOMC of \Propagate{$\phi$, $\Delta$,
  $n$}, assuming $|\Gamma| = m$, is the same as the FOMC of \Propagate{$\phi$,
  $\Gamma$, $m$}, assuming $|\Delta| = n$.

Although \CompileWithBaseCases starts with a call to \textsc{Crane}, the
proposed algorithm is not merely a post-processing step for FOKC because
\cref{alg:compilewithbasecases} is recursive and can issue additional calls to
\textsc{Crane} on various derived sentences. We conclude this section by
bounding the recursion depth of the \CompileWithBaseCases procedure, thereby
also proving that the algorithm terminates.

\begin{theorem}
  Let $\phi$ be a sentence with $n$ domains. The maximum recursion depth of
  \CompileWithBaseCases{$\phi$} is then $n$.
\end{theorem}

The proof of this theorem relies on two observations regarding the algorithms
presented in \cref{sec:identifying,sec:propagating}.

\begin{observation}\label{observation1}
  Each base case returned by \FindBaseCases contains at least one constant (in
  line with \cref{def:basecase}).
\end{observation}

\begin{observation}\label{observation2}
  For any sentence $\phi$, domain $\Delta$, and $n \in \mathbb{N}_{0}$,
  \Propagate{$\phi$, $\Delta$, $n$} returns a sentence with no variables
  quantified over $\Delta$.
\end{observation}

\begin{proof}
  We proceed by induction on $n$. If $n=0$, then $\phi$ is essentially a
  propositional formula, and \textsc{Crane} compiles it into an equation of the
  form $f = \expr$ with no `function calls'. Suppose that---for all sentences
  with at most $n$ domains---\CompileWithBaseCases terminates with a recursion
  depth of at most $n$. Let $\phi$ be a sentence with $n+1$ domains. By
  \cref{observation1}, each base case on \autoref{line:loopy} of
  \cref{alg:compilewithbasecases} has at least one constant. Therefore, by
  \cref{observation2}, after \autoref{line:propagate}, the sentence $\psi$ has
  at most $n$ domains. Thus, \autoref{line:final} terminates with a recursion
  depth of at most $n$ by the inductive hypothesis, completing the inductive
  proof.
\end{proof}

\subsection{Identifying a Sufficient Set of Base Cases}\label{sec:identifying}

\begin{algorithm}[t]
  \caption{\protect\FindBaseCases{$\mathcal{E}$}}\label{alg:findbasecases}
  \KwIn{set $\mathcal{E}$ of equations}
  \KwOut{set $\mathcal{B}$ of base cases}

  $\mathcal{B} \gets \emptyset$\;
  \ForEach{function call $f(\mathbf{y})$ on the RHS of an equation in $\mathcal{E}$}{
    $\mathbf{x} \gets \text{the parameters of $f$ in its definition}$\;
    \ForEach{$y_{i} \in \mathbf{y}$} {
      \lIf{$y_{i} \in \mathbb{N}_{0}$}{$\mathcal{B} \gets \mathcal{B} \cup \{\, f(\mathbf{x})[x_{i} \mapsto y_{i}] \,\}$}\label{line:constant}
      \ElseIf{$y_{i} = x_{i} - c_{i}$ for some $c_{i} \in \mathbb{N}_{0}$}{
        \lFor{$j \gets 0$ \KwTo $c_{i} - 1$}{$\mathcal{B} \gets \mathcal{B} \cup \{\, f(\mathbf{x})[x_{i} \mapsto j] \,\}$}\label{line:nonconstant}
      }
    }
  }
\end{algorithm}

\Cref{alg:findbasecases} summarises the implementation of \FindBaseCases. It
considers two types of arguments when a function $f$ calls itself recursively:
constants and arguments of the form $x_{i} - c_{i}$. When the argument is a
constant $c_{i}$, \autoref{line:constant} adds a base case with $c_{i}$ to the
set of all base cases $\mathcal{B}$. In the second case,
\autoref{line:nonconstant} adds a base case to $\mathcal{B}$ for each constant
from $0$ to (but not including) $c_{i}$.

\begin{example}
  Consider the recursive function $g$ from \cref{example:solution}.
  \FindBaseCases iterates over two function calls: $g(l-1, m)$ and
  $g(l-1, m-1)$. The former produces the base case $g(0, m)$, while the latter
  produces both $g(0, m)$ and $g(l, 0)$.
\end{example}

In the rest of this section, we will show that the base cases identified by
\FindBaseCases are sufficient for the algorithm to terminate.\footnote{Note that
  characterising the fine-grained complexity of the solutions found by
  \Cranetwo{} or other FOMC algorithms is an emerging area of research. These
  questions have been partially addressed in previous
  work~\cite{DBLP:conf/kr/DilkasB23,DBLP:conf/kr/TothK24} and are unrelated to
  the goals of this section.}

\begin{theorem}\label{thm:halting}
  Let $\mathcal{E}$ denote the set of equations returned by
  \CompileWithBaseCases, and let $f$ be one of the functions defined in
  $\mathcal{E}$. The evaluation of $f(\mathbf{x})$ terminates for any tuple
  $\mathbf{x}$ of non-negative integers (of appropriate length).
\end{theorem}

We prove \cref{thm:halting} using double induction. First, we apply induction to
the number of functions in $\mathcal{E}$. Then, we use induction on the arity of
the `last' function in $\mathcal{E}$ according to a topological ordering. Before
proving \cref{thm:halting}, we make a few observations that are immediate
consequences of this and previous
work~\cite{DBLP:conf/kr/DilkasB23,DBLP:conf/ijcai/BroeckTMDR11}.

\begin{observation}\label{assumption1}
  For each function $f$, there is precisely one equation $e \in \mathcal{E}$
  with $f(\mathbf{x})$ on the LHS where all $x_{i}$'s are variables (i.e., $e$
  is not a base case). We refer to $e$ as the \emph{definition} of $f$.
\end{observation}

\begin{observation}\label{assumption2}
  There is a \emph{topological ordering} ${(f_{i})}_{i}$ of all functions in
  $\mathcal{E}$ such that equations in $\mathcal{E}$ with $f_{i}$ on the LHS do
  not contain function calls to $f_{j}$ with $j > i$.
\end{observation}

\Cref{assumption2} prevents mutual recursion and other cyclic scenarios. It
arises from the timeline of FOKC and the fact that, while constructing a graph,
the algorithm cannot draw an edge to a vertex that does not yet exist.

\begin{observation}\label{assumption3}
  For each equation $(f(\mathbf{x}) = \expr) \in \mathcal{E}$, the evaluation
  of $\expr$ terminates when provided with the values of all relevant function
  calls.
\end{observation}

Indeed, each $\expr$ is defined by a node type, all of which are finite
computations listed in previous
work~\cite{DBLP:conf/kr/DilkasB23,DBLP:conf/nips/Broeck11,DBLP:conf/ijcai/BroeckTMDR11}.

\begin{corollary}\label{fact}
  If $f$ is a non-recursive function with no function calls on the RHS of its
  definition, then the evaluation of any function call $f(\mathbf{x})$
  terminates.
\end{corollary}

\begin{observation}\label{fact2}
  For each equation $(f(\mathbf{x}) = \expr{}) \in \mathcal{E}$, if $\mathbf{x}$
  contains only constants, then $\expr{}$ cannot include any function calls to
  $f$.
\end{observation}

Additionally, we introduce an assumption about the structure of recursion.

\begin{assumption}\label{assumption4}
  For each equation $(f(\mathbf{x}) = \expr) \in \mathcal{E}$, every recursive
  function call $f(\mathbf{y}) \in \expr$ satisfies the following:
  \begin{itemize}
    \item each $y_{i}$ is either $x_{i} - c_{i}$ or $c_{i}$ for some constant
          $c_{i}$ and
    \item there exists $i$ such that $y_{i} = x_{i} - c_{i}$ for some
          $c_{i} > 0$.
  \end{itemize}
\end{assumption}

Finally, we assume a particular order of evaluation for function calls using the
equations in $\mathcal{E}$; specifically, base cases precede the recursive
definition.

\begin{assumption}
  When multiple equations in $\mathcal{E}$ match a function call
  $f(\mathbf{x})$, preference is given to the equation with the most constants
  on its LHS.
\end{assumption}

With the observations and assumptions mentioned above, we are ready to prove
\cref{thm:halting}. For readability, we divide the proof into several lemmas of
increasing generality.

\begin{lemma}\label{lemma:oneunary}
  Assume that $\mathcal{E}$ consists of just \emph{one unary} function called
  $f$. Then the evaluation of $f(x)$ terminates for any $x \in \mathbb{N}_{0}$.
\end{lemma}
\begin{proof}
  If $f(x)$ matches a base case, then its evaluation terminates by
  \cref{fact,fact2}. If $f$ is not recursive, the evaluation of $f(x)$ also
  terminates by \cref{fact}.

  Otherwise, let $f(y)$ be an arbitrary function call on the RHS of the
  definition of $f(x)$. If $y$ is a constant, then there is a base case for
  $f(y)$. Otherwise, let $y = x - c$ for some $c > 0$. Then there exists
  $k \in \mathbb{N}_{0}$ such that $0 \le x - kc \le c-1$. So, after $k$
  iterations, the sequence of function calls $f(x), f(x-c), f(x-2c),\dots$ will
  match the base case $f(x \mod c)$.
\end{proof}

\begin{lemma}\label{lemma:onefunction}
  Generalising \cref{lemma:oneunary}, let $\mathcal{E}$ be a set of equations
  for \emph{one} $n$-ary function $f$ for some $n \ge 1$. The evaluation of
  $f(\mathbf{x})$ terminates for any $\mathbf{x} \in \mathbb{N}_{0}^{n}$.
\end{lemma}
\begin{proof}
  If $f$ is non-recursive, the evaluation of $f(\mathbf{x})$ terminates by
  previous arguments. We proceed by induction on $n$, with the base case of
  $n=1$ handled by \cref{lemma:oneunary}. Assume that $n > 1$. Any base case of
  $f$ can be seen as a function of arity $n-1$ since one of the parameters is
  constant. Thus, the evaluation of any base case terminates by the inductive
  hypothesis. It remains to show that the evaluation of the recursive equation
  for $f$ terminates, but this follows from \cref{assumption3}.
\end{proof}

Now that we have handled the special case when $\mathcal{E}$ defines one
function, we are ready to prove the main theorem.

\begin{proof}[Proof of \cref{thm:halting}]
  We proceed by induction on the number of functions $n$. The base case of $n=1$
  is handled by \cref{lemma:onefunction}. Let ${(f_{i})}_{i=1}^{n}$ be some
  topological ordering of these $n > 1$ functions. If $f = f_{j}$ for $j < n$,
  then the evaluation of $f(\mathbf{x})$ terminates by the inductive hypothesis
  since $f_{j}$ cannot call $f_{n}$ by \cref{assumption2}. Using the inductive
  hypothesis that all function calls to $f_{j}$ (with $j < n$) terminate, the
  proof proceeds similarly to the proof of \cref{lemma:onefunction}.
\end{proof}

\subsection{Propagating Domain Size Assumptions}\label{sec:propagating}

\begin{algorithm}[t]
  \caption{\protect\Propagate{$\phi$, $\Delta$, $n$}}\label{alg:propagate}
  \KwIn{sentence $\phi$, domain $\Delta$, $n \in \mathbb{N}_{0}$}
  \KwOut{sentence $\phi'$}
  $\phi' \gets \emptyset$\;
  \uIf{$n = 0$}{
    \ForEach{clause $C \in \phi$}{
      \lIf{$\Delta \not\in \Doms(C)$}{$\phi' \gets \phi' \cup \{\, C \,\}$}
      \Else{
        $C' \gets \{\, l \in C \mid \Delta \not\in \Doms(l) \,\}$\;
        \If{$C' \ne \emptyset$}{\label{line:presmoothing}
          $l \gets \text{an arbitrary literal in } C'$\;\label{line:smoothing1}
          $\phi' \gets \phi' \cup \{\, C' \cup \{\, \neg l \,\} \,\} $\;\label{line:smoothing2}
        }
      }
    }
  }
  \Else{
    $D \gets \text{a set of $n$ new constants in $\Delta$}$\;
    \ForEach{clause $C \in \phi$}{
      ${(x_{i})}_{i=1}^{m} \gets \text{the variables in $C$ with domain $\Delta$}$\;
      \lIf{$m = 0$}{$\phi' \gets \phi' \cup \{\, C \,\}$}
      \lElse{$\phi' \gets \phi' \cup \{\, C[x_{1} \mapsto c_{1}, \dots, x_{m} \mapsto c_{m}] \mid {(c_{i})}_{i=1}^{m} \in D^{m} \,\}$}
    }
  }
\end{algorithm}

\Cref{alg:propagate}, called \Propagate, modifies the sentence $\phi$ based on
the assumption that $|\Delta| = n$. When $n=0$, some clauses become vacuously
satisfied and can be removed. When $n > 0$, partial grounding replaces all
variables with domain $\Delta$ with constants. (None of the sentences examined
in this work had $n > 1$.) \Cref{alg:propagate} handles these two cases
separately. For a literal or clause $C$, we write the set of corresponding
domains as $\Doms(C)$. In the case of $n = 0$, there are three types of clauses
to consider:
\begin{enumerate}
  \item those that do not mention $\Delta$,\label[type]{type1}
  \item those in which every literal contains variables quantified over
  $\Delta$ and\label[type]{type2}
  \item those with some literals containing such variables and some
        without.\label[type]{type3}
\end{enumerate}
We transfer clauses of \cref{type1} to the new sentence $\phi'$ without any
changes. For clauses of \cref{type2}, $C'$ is empty, so these clauses are
filtered out. As for clauses of \cref{type3},
lines~\ref{line:presmoothing}--\ref{line:smoothing2} perform a new kind of
smoothing, the explanation of which we defer to \cref{sec:smoothing}.

In the case of $n>0$, $n$ new constants are introduced. Let $C$ be an arbitrary
clause in $\phi$, and let $m \in \mathbb{N}_{0}$ be the number of variables in
$C$ quantified over $\Delta$. If $m=0$, $C$ is added directly to $\phi'$.
Otherwise, a clause is added to $\phi'$ for every possible combination of
replacing the $m$ variables in $C$ with the $n$ new constants.

\begin{example}
  Let $C \coloneqq \forall x \in \Gamma\text{. }\forall y, z \in \Delta\text{.
  } \neg P(x, y) \lor \neg P(x, z) \lor y=z$. Then
  $\Doms(C) = \Doms(\neg P(x, y)) = \Doms(\neg P(x, z)) = \{\, \Gamma, \Delta \,\}$,
  and $\Doms(y=z) = \{\, \Delta \,\}$. A call to \Propagate{$\{\, C \,\}$,
    $\Delta$, $3$} would result in the following sentence with nine clauses:
  \begin{align*}
    (\forall x \in \Gamma\text{. }\neg P(x, c_{1}) \lor& \neg P(x, c_{1}) \lor c_{1}=c_{1})\land{}\\
    (\forall x \in \Gamma\text{. }\neg P(x, c_{1}) \lor& \neg P(x, c_{2}) \lor c_{1}=c_{2})\land{}\\
    \vdots&\\
    (\forall x \in \Gamma\text{. }\neg P(x, c_{3}) \lor& \neg P(x, c_{3}) \lor c_{3}=c_{3}).
  \end{align*}
  Here, $c_{1}$, $c_{2}$, and $c_{3}$ are the new constants.
\end{example}

\subsection{Smoothing the Base Cases}\label{sec:smoothing}

Smoothing modifies a circuit to reintroduce eliminated atoms, ensuring the
correct model count~\cite{darwiche2001tractable,DBLP:conf/ijcai/BroeckTMDR11}.
This section describes a similar process performed on
lines~\ref{line:presmoothing}--\ref{line:smoothing2} of \cref{alg:propagate}.
Line~\ref{line:presmoothing} checks if smoothing is necessary, and
lines~\ref{line:smoothing1} and~\ref{line:smoothing2} execute it. If the
condition on \autoref{line:presmoothing} is not satisfied, the clause is not
smoothed but omitted.

Suppose \CompileWithBaseCases calls \Propagate with arguments
$(\phi, \Delta, 0)$, i.e., we are simplifying the sentence $\phi$ by assuming
that the domain $\Delta$ is empty. Informally, if there is a predicate $P$ in
$\phi$ unrelated to $\Delta$, smoothing preserves all occurrences of $P$, even
if all clauses with $P$ become vacuously satisfied.

\begin{example}\label{example:basecasesmoothing}
  Let $\phi$ be
  \begin{align}
    (\forall x \in \Delta\text{. }\forall y, z \in \Gamma&\text{. }Q(x) \lor P(y, z))\land{}\label[clause]{eq:example1}\\
    (\forall y, z \in \Gamma'&\text{. }P(y, z))\label[clause]{eq:example2},
  \end{align}
  where $\Gamma' \subseteq \Gamma$ is a domain introduced by a compilation rule.
  Note that $P$, as a relation, is a subset of $\Gamma \times \Gamma$.

  Now, let us reason manually about the model count of $\phi$ when
  $\Delta = \emptyset$. Predicate $Q$ can only take one value, $Q = \emptyset$.
  The value of $P$ is fixed over $\Gamma' \times \Gamma'$ by \cref{eq:example2},
  but it can vary freely over
  $(\Gamma \times \Gamma) \setminus (\Gamma' \times \Gamma')$ since
  \cref{eq:example1} is vacuously satisfied by all structures. Therefore, the
  correct FOMC should be $2^{|\Gamma|^2 - |\Gamma'|^2}$. However, without
  \autoref{line:smoothing2}, \Propagate would simplify $\phi$ to
  $\forall y, z \in \Gamma'\text{. }P(y, z)$. In this case, $P$ is a subset of
  $\Gamma' \times \Gamma'$. This simplified sentence has only one model:
  $\{\, P(y, z) \mid y, z \in \Gamma' \,\}$. By including
  \autoref{line:smoothing2}, \Propagate transforms $\phi$ to
  \[
    (\forall y, z \in \Gamma\text{. }P(y, z) \lor \neg P(y, z)) \land (\forall y, z \in \Gamma'\text{. }P(y, z)),
  \]
  which retains the correct model count.
\end{example}

It is worth mentioning that the choice of $l$ on \autoref{line:smoothing1} of
\cref{alg:propagate} is inconsequential because any choice achieves the same
goal: constructing a tautological clause that retains the literals in $C'$.

\subsection{Generating C++ Code}\label{sec:cpp}

In this section, we will describe the final step of \Cranetwo{} as outlined in
\cref{fig:overview}, i.e., translating the set of equations $\mathcal{E}$ into
C++ code. (Note that the compilation does not rely on any unique features of
C++; the equations could easily be compiled into any other programming
language.) Recall that this step is crucial for the usability of the algorithm;
otherwise, function definitions would remain purely mathematical, with no
convenient way to compute the model count for particular domain sizes. Once a
C++ program is produced, it can be executed with different command-line
arguments to determine the model count of the sentence for various domain sizes.
A typical C++ program for one of the benchmarks in \cref{sec:experiments}
contains about \numrange{80}{100} lines of code. Moreover, although we focus on
FOMC in this work, support for WFOMC could easily be added by using
floating-point numbers instead of integers.

\begin{algorithm}[t]
  \caption{A sketch of the C++ program for the equations in
    \cref{example:solution}, particularly highlighting the recursive definition
    of the function $g$.}\label{alg:cpp}
  initialise $\Cache_{g(0, m)}$, $\Cache_{g(l, 0)}$, $\Cache_{g}$, and
  $\Cache_{f}$\; \DontPrintSemicolon \textbf{Function} $g_{0,m}(m)$\textbf{:}
  \dots\; \textbf{Function} $g_{l,0}(l)$\textbf{:} \dots\; \PrintSemicolon
  \Fn{$g(l, m)$}{
    \lIf{$(l, m) \in \Cache_{g}$}{\Return{$\Cache_{g}(l, m)$}}\label{line:cache}
    \lIf{$l = 0$}{\Return{$g_{0,m}(m)$}}\label{line:basecaseone}
    \lIf{$m = 0$}{\Return{$g_{l,0}(l)$}}\label{line:basecasetwo}
    $r \gets g(l-1, m) + mg(l-1,m-1)$\; $\Cache_{g}(l, m) \gets r$\;
    \Return{$r$}\; } \DontPrintSemicolon \textbf{Function} $f(m, n)$\textbf{:}
  \dots\; \PrintSemicolon \Fn{\Main}{ $(m, n) \gets \Input{}$\;
    \Return{$f(m, n)$}\; }
\end{algorithm}

See \cref{alg:cpp} for the typical structure of a generated C++ program. Each
equation in $\mathcal{E}$ turns into a C++ function with a separate cache for
memoisation. Hence, \cref{alg:cpp} has a function and a cache for
$f(\cdot, \cdot)$, $g(\cdot, \cdot)$, $g(\cdot, 0)$, and $g(0, \cdot)$. The
implementation of an equation consists of three parts. First (on
\autoref{line:cache}), we check if the arguments already exist in the
corresponding cache. If they do, we return the cached value. Second (on
lines~\ref{line:basecaseone} and~\ref{line:basecasetwo}), we check if the
arguments match any of the base cases (as defined in \cref{sec:algebra}). If so,
we redirect the arguments to the C++ function for that base case. Finally, if
none of the above cases applies, we evaluate the arguments based on the
expression on the RHS of the equation, store the result in the cache, and return
it.

\section{Experimental Evaluation}\label{sec:experiments}

Our empirical evaluation sought to compare the runtime performance of
{\Cranetwo} with the current state of the art, namely
\textsc{FastWFOMC}\footnote{\url{https://github.com/jan-toth/FastWFOMC.jl}}~\cite{DBLP:conf/kr/TothK24,DBLP:conf/uai/BremenK21}
and \textsc{ForcLift}\footnote{\url{https://github.com/UCLA-StarAI/Forclift}}.
Our experiments involved two versions of \Cranetwo{}: \Cranegreedy{} and
\Cranebfs{}. Like its predecessor (see \cref{sec:compilation}), \Cranetwo{} has
two modes for applying compilation rules to sentences: one that uses a greedy
search algorithm similar to \textsc{ForcLift} and another that combines greedy
and BFS\@.

The experiments used an Intel Skylake \SI{2.4}{\giga\hertz} CPU with
\SI{188}{\gibi\byte} of memory and CentOS~7. We used the Intel C++ Compiler
2020u4 for C++ programs, Julia~1.10.4 for \textsc{FastWFOMC}, and the Java
Virtual Machine 1.8.0\_201 for \textsc{ForcLift} and \Cranetwo{}. Although
implemented in different languages, \Cranetwo{} and \textsc{FastWFOMC} use the
same GNU Multiple Precision Arithmetic Library for arbitrary-precision
arithmetic. Additionally, while \textsc{ForcLift} could be extended to support
arbitrary-precision arithmetic, its compilation rules are not extensive enough
to handle some of the benchmarks.

We ran each algorithm on each benchmark using domains of size
$2^{1}, 2^{2}, 2^{3}$, and so on until an algorithm failed to handle a domain
size due to a timeout (of one hour) or out-of-memory or out-of-precision errors.
We verified the accuracy of the numerical answers using the corresponding
integer sequences in the On-Line Encyclopedia of Integer Sequences~\cite{oeis}.

\subsection{Benchmarks}

We compare these algorithms using three benchmarks from previous work. The first
benchmark is the bijection-counting problem from \cref{example:bijections}. The
second benchmark is a variant of the well-known \friends{} Markov logic
network~\cite{DBLP:conf/aaai/SinglaD08,DBLP:conf/uai/BroeckCD12}, which takes
the form of
\[
  (\forall x,y \in \Delta\text{.
  } S(x) \land F(x, y) \Rightarrow S(y)) \land (\forall x \in \Delta\text{.
  }S(x) \Rightarrow C(x)).
\]
In this sentence, we have three predicates, $S$, $F$, and $C$, that denote
smoking, friendship, and cancer, respectively. The first clause states that
friends of smokers are also smokers, and the second clause asserts that smoking
causes cancer. Common additions to this sentence include making the friendship
relation symmetric and assigning probabilities to each clause. Finally, we
include the function-counting problem~\cite{DBLP:conf/kr/DilkasB23}
\begin{equation}\label[formula]{eq:functions}
  (\forall x \in \Gamma\text{. }\exists y \in \Delta\text{.
  }P(x, y)) \land (\forall x \in \Gamma\text{. }\forall y, z \in \Delta\text{.
  }P(x, y) \land P(x, z) \Rightarrow y = z)
\end{equation}
as our third benchmark. Here, predicate $P$ represents a function from $\Gamma$
to $\Delta$. The first clause asserts that each $x$ must have at least one
corresponding $y$, while the second clause ensures the uniqueness of such a $y$.

\begin{remark*}
  We formulate the \bijections{} and \functions{} benchmarks using two domains,
  $\Gamma$ and $\Delta$, as this formulation is known to help FOKC algorithms
  find efficient solutions~\cite{DBLP:conf/kr/DilkasB23}. To compare \Cranetwo{}
  and \textsc{ForcLift} with \textsc{FastWFOMC}, which has no support for
  multiple domains, we set $|\Gamma| = |\Delta|$.
\end{remark*}

\paragraph*{Our Benchmarks in \Ctwo{} and \UFO{}}
For completeness and reproducibility, let us translate the benchmark sentences
from \FO{} to \Ctwo{} and \UFO{}. Since \friends{} is a relatively simple
sentence, it remains the same in \Ctwo{} and \UFO{}. For \functions{}, in
\Ctwo{}, one would write
\[
  \forall x \in \Delta\text{. }\exists^{=1} y \in \Delta\text{. }P(x, y).
\]
In \UFO{}, the equivalent formulation is
\begin{equation}\label[formula]{eq:functions1}
  (\forall x, y \in \Delta\text{. }S(x) \lor \neg P(x, y)) \land (|P| = |\Delta|),
\end{equation}
where $w^{-}(S) = -1$. Although \cref{eq:functions1} has more models than its
counterpart in \Ctwo{}, the negative weight $w^{-}(S) = -1$ causes some of the
terms in the definition of WFOMC to cancel out. The translation of \bijections{}
is similar to that of \functions{}. In \Ctwo{}, one could write
\[
  (\forall x \in \Delta\text{. }\exists^{=1} y \in \Delta\text{.
  }P(x, y)) \land (\forall y \in \Delta\text{. }\exists^{=1} x \in \Delta\text{.
  }P(x, y)).
\]
Similarly, in \UFO{}, the equivalent formulation is
\[
  (\forall x, y \in \Delta\text{.
  }R(x) \lor \neg P(x, y)) \land (\forall x, y \in \Delta\text{.
  }S(x) \lor \neg P(y, x)) \land (|P| = |\Delta|),
\]
where $w^{-}(R) = w^{-}(S) = -1$.

The three benchmarks cover a wide range of possibilities. The \friends{}
benchmark uses multiple predicates and only two variables without employing any
constructs specific to only one of the (W)FOMC algorithms. The \functions{}
benchmark, on the other hand, can still be handled by all the algorithms, but
its formulation relies on algorithm-specific constructs. Namely, for
\textsc{ForcLift} and \Cranetwo{}, the sentence is written using three
variables. Since \textsc{FastWFOMC} supports at most two variables, we rewrite
the \cref{eq:functions} with two variables and a \emph{cardinality constraint}.
Lastly, the \bijections{} benchmark is an example of a sentence that
\textsc{FastWFOMC} can handle but \textsc{ForcLift} cannot.

\subsection{Results}

\begin{figure}[t]
  \centering
  \includegraphics{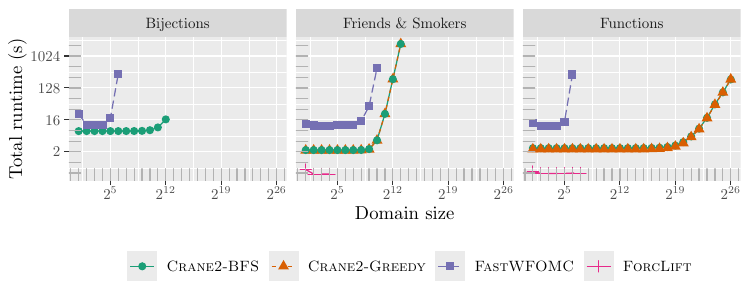}
  \caption{The (total) runtime of the algorithms as a function of domain size.
    Note that both axes are on a logarithmic scale.}\label{fig:plot}
\end{figure}

\Cref{fig:plot} presents a summary of the experimental results. Only
\textsc{FastWFOMC} and \Cranebfs{} could handle the bijection-counting problem.
For this benchmark, the largest domain sizes these algorithms could accommodate
were \num{64} and \num{4096}, respectively. On the other two benchmarks,
\textsc{ForcLift} had the lowest runtime. However, since it can only handle
model counts smaller than $2^{31}$, it only scales up to domain sizes of
\num{16} and \num{128} for \friends{} and \functions{}, respectively.
\textsc{FastWFOMC} outperformed \textsc{ForcLift} in the case of \friends{}, but
not \functions{}, as it could handle domains of size \num{1024} and \num{64},
respectively. Furthermore, both \Cranebfs{} and \Cranegreedy{} performed
similarly on both benchmarks. Similarly to the \bijections{} benchmark,
\Cranetwo{} significantly outperformed the other two algorithms, scaling up to
domains of size \num{8192} and \num{67108864}, respectively.

\begin{table}[t]
  \caption{The compilation times of FOKC algorithms, sorted in ascending
    order.}\label{tbl:compilation}
  \centering
  \begin{tabular}{llr}
    \toprule
    Algorithm & Benchmark & Compilation time (s)\\
    \midrule
    \rowcolor{gray!10}
              & \functions{} & 0.44\\
    \rowcolor{gray!10}
    \multirow{-2}{*}{\textsc{ForcLift}} & \friends{} & 0.52\\
              & \friends{} & 2.20\\
    \multirow{-2}{*}{\Cranegreedy{}} & \functions{} & 2.36\\
    \rowcolor{gray!10}
              & \friends{} & 2.23\\
    \rowcolor{gray!10}
              & \functions{} & 2.53\\
    \rowcolor{gray!10}
    \multirow{-3}{*}{\Cranebfs{}} & \bijections{} & 7.70\\
  \end{tabular}
\end{table}

One might notice that the runtime of \textsc{FastWFOMC} and \textsc{ForcLift} is
slightly higher for the smallest domain size. This peculiarity is the
consequence of \emph{just-in-time} (JIT) compilation. As \Cranetwo{} is only run
once per benchmark, we include the JIT compilation time in its overall runtime
across all domain sizes. Additionally, \cref{tbl:compilation} summarises the
compilation time of each FOKC algorithm on each suitable benchmark. While the
compilation time of \textsc{ForcLift} is generally lower compared to
\Cranetwo{}, neither significantly affects the overall runtime.

Based on our experiments, which algorithm should one use in practice? If
\textsc{ForcLift} can handle the sentence and the domain sizes are reasonably
small, it is likely the fastest algorithm. In other situations, \Cranetwo{} will
likely be significantly more efficient than \textsc{FastWFOMC} regardless of
domain size, provided both algorithms can handle the sentence.

\section{Conclusion and Future Work}\label{sec:conclusion}

In this work, we have presented a scalable, automated FOKC-based approach to
FOMC\@. Our algorithm involves completing the definitions of recursive functions
and subsequently translating all function definitions into C++ code. Empirical
results demonstrate that \Cranetwo{} can scale to larger domain sizes than
\textsc{FastWFOMC} while supporting a wider range of sentences than
\textsc{ForcLift}. The ability to efficiently handle large domain sizes is
particularly crucial in the weighted setting, as illustrated by the \friends{}
example, where the model captures complex social networks with probabilistic
relationships. Without this scalability, these models would have limited
practical value.

Future directions for research include conducting a comprehensive experimental
comparison of FOMC algorithms to better understand their comparative performance
across various sentences. The capabilities of \Cranetwo{} could also be
characterised theoretically, for example, by proving completeness for logic
fragments liftable by other
algorithms~\cite{DBLP:journals/jair/Kuzelka21,DBLP:conf/aaai/TothK23,DBLP:journals/ai/BremenK23}.
Additionally, the efficiency of FOMC algorithms can be further analysed using
fine-grained complexity, which would provide more detailed insights into the
computational demands of different sentences.

\bibliography{paper}

\end{document}